\documentclass[]{informs3} 
\OneAndAHalfSpacedXII 

\newenvironment{proof}{\emph{Proof: }}{\hfill$\square$}
\usepackage{natbib}
\usepackage[ruled]{algorithm}
\usepackage[noend]{algorithmic}
 \bibpunct[, ]{(}{)}{,}{a}{}{,}%
 \def\bibfont{\small}%
\allowdisplaybreaks 

\newcommand{\CUT}[1]{{}}

\renewcommand{\paragraph}[1]{\smallskip \noindent {\textbf{#1}}}

\input epsf

\newtheorem{theorem}{Theorem}

\newtheorem{assumption}{Assumption}

\newtheorem{lemma}{Lemma}

\begin{document}




 \TITLE{Stochastic Submodular Probing with State-Dependent Costs}

\ARTICLEAUTHORS{%
\AUTHOR{Shaojie Tang}
\AFF{University of Texas at Dallas}
} 

\ABSTRACT{%
In this paper, we study a new stochastic submodular maximization problem with state-dependent costs and rejections. The input of our problem is a budget constraint $B$, and a set of items whose states (i.e., the marginal contribution and the cost of an item) are drawn from a known probability distribution. The only way to know the realized state of an item is to probe that item. We allow rejections, i.e., after probing an item and knowing its actual state, we must decide immediately and irrevocably whether to add that item to our solution or not. Our objective is to sequentially probe/selet a best group of items  subject to a budget constraint on the total cost of the selected items. We present a constant approximate solution to this problem. We show that our solution can be extended to an online setting.}


\maketitle

%

\section{Introduction}
\label{sec:intro}
In this paper, we study a new stochastic submodular maximization problem. We introduce the state-dependent item costs and rejections into the classic stochastic submodular maximization problem. The input of our problem is a budget constraint $B$, and a set of items whose states  are drawn from a known probability distribution. The marginal contribution and the cost of an item is dependent on its actual state. We must probe an item in order to reveal its actual state. After probing an item and knowing its actual state, one must decide immediately and irrevocably whether to add that item to our solution or not. Our objective is to sequentially probe/select a best group of items subject to a budget constraint on the total cost of the selected items. We present a constant approximate solution to this problem. Perhaps surprisingly, our algorithm also applies to an online setting described as follows: suppose there is a sequence of items arriving in  an adversarial order, on the arrival of an item, we must decide immediately and irrevocably whether to select it or not after seeing its realization. For this online decision problem, our algorithm achieves  the same approximation ratio as obtained under the offline setting.

\emph{Related works.} Stochastic submodular maximization has been extensively studied recently \cite{golovin2011adaptive,chen2013near,fujii2016budgeted}. However, most of existing works assume that the cost of an item is deterministic and pre-known. We relax this assumption by introducing the state-dependent item cost. In particular, we assume that the actual cost of an item is decided by its realized state. We must probe an item in order to know its state. When considering linear objective function, our problem reduces to the stochastic knapsack problem with rejections \cite{gupta2011approximation}.  \cite{gupta2011approximation} gave a constant approximate algorithm for this problem. Recently, \cite{fukunaga2019stochastic} studied the stochastic submodular maximization problem with performance-dependent costs, however, their model does not allow rejections. Therefore, our problem does not coincide with their work. Moreover, it is not immediately clear how to extend their algorithm to online setting.
Our work is also closely related to submodular probing problem \cite{adamczyk2016submodular}, however, they assume each item  has only two states, i.e., active or inactive, we relax this assumption by allowing each item to have multiple states and the item cost is dependent on its state. Furthermore, their model does not allow rejections, i.e., one can not reject  an active item after it has been probed.


\section{Preliminaries and Problem Formulation}
\label{sec:behavior}
 \paragraph{Lattice-submodular functions} Let $[I]=\{1,2,\cdots, I\}$ be a set of items and $[S]=\{1,2,\cdots, S\}$ be a set of states.  Given two vectors $u,v \in [S]^{[I]}$,  $u\leq v$ means that $u(i)\leq v(i)$ for all $i\in [I]$. Define $(u\vee v)(i)=\max\{u(i),v(i)\}$ and $(u\wedge v)(i)=\min\{u(i),v(i)\}$. For each $i\in [I]$, define $\mathbf{1}_i$ as the vector that has a $1$ in the $i$-th coordinate and $0$ in all other coordinates.  A function $f: [S]^{[I]}\rightarrow \mathbb{R}_{+}$ is called \emph{monotone} if $f(u)\leq f(v)$ holds for any $u, v\in [S]^{[I]}$ such that $u\leq v$, and $f$ is called \emph{lattice submodular} if $f(u\vee s\mathbf{1}_i)-f(u) \geq f(v\vee s\mathbf{1}_i)-f(v)$ holding for any $u,v\in [S]^{[I]}$, $s\in [S]$, $i\in [I]$.

 \paragraph{Items and States} We let vector $\Phi\in [S]^{[I]}$ denote random states of all items. For each item $i\in [I]$, let $\Phi(i)\in [S]$ denote the random state of  item $i$.  Let $\phi(i)$ denote a \emph{realization} of $\Phi(i)$. The state of each item is unknown initially, one must probe an item before observing its realization. We allow rejections, i.e., after probing an item and knowing its  state, we must decide immediately and irrevocably whether to pick that item  or not. We assume there is a known prior probability distribution $\mathcal{D}_i$ over realizations for each item $i\in[I]$, i.e., $\mathcal{D}_i=\{\Pr[\Phi = \phi]: \phi \in [S]^{[I]}\}$.  The states of all items are decided independently at random, i.e., $\phi$ is drawn randomly from the product distribution $\mathcal{D}=\prod_{i\in [I]}\mathcal{D}_i$. For each $(i,s) \in[I]\times [S]$, we use $c_i(s)$ to denote the cost of an item $i$ when its state is $s$.

\begin{assumption}
\label{asum:1}We  assume that $c_i(s)\geq c_i(s')$ for any $i\in I$ and $s, s'\in [S]$ such that $s \geq s'$, i.e., the cost of an item is larger if it is in a ``better'' state.
\end{assumption}

The above assumption can also be found in \cite{fukunaga2019stochastic}. For each set of item-state pairs $U\subseteq [I]\times [S]$, we define a vector $u\in [S]^{[I]}$ such that $u(i)=0$ if $(i,s)\notin U$, otherwise $u(i)=\max\{s\mid (i,s)\in U\}$. Now we are ready to introduce a set function $h$ over a new ground set $[I]\times [S]$: consider an arbitrary set of item-state pairs $U\subseteq [I]\times [S]$, define $h(U)=f(u)$. It is easy to verify that if $f$ is monotone and lattice-submodular, then  $h$ is monotone and submodular. Given an $I \times S$ matrix $\mathbf{x}$, we define the multilinear extension $H$ of $h$ as:
\[H(\mathbf{x})=\sum_{U\subseteq  [I]\times [S]} h(U) \prod_{(i,s)\in U} x_{is}\prod_{(i,s)\notin U} (1-x_{is})\]
The value $H(\mathbf{x})$ is  the expected value of $h(R)$ where $R$ is a random set obtained by picking each element $(i,s)\in  [I]\times [S]$ independently with probability $x_{is}$.

\paragraph{Adaptive Policy and Problem Formulation}  We model the adaptive strategy of probing/picking items through a policy $\pi$. Formally, a policy $\pi$ is a function that specifies which item to probe/pick next based on the observations made so far. Consider an arbitrary policy $\pi$, assume that conditioned on  $\Phi=\phi$,  $\pi$ picks a set of items (and corresponding states) $G(\pi, \phi)\subseteq [I]\times[S]$\footnote{For simplicity, we only consider deterministic policy. However, all results can be easily extended to random policies.}.
The expected utility of $\pi$ is $f(\pi)=\sum_{\phi}\Pr[\Phi=\phi] h(G(\pi, \phi))$. We say a policy $\pi$ is \emph{feasible} if for any $\phi$ such that $\Pr[\Phi = \phi]>0$, $\sum_{(i,s)\in G(\pi, \phi)}c_i(s)\leq B$ where $B$ a budget constraint.  Our goal is to identify the best feasible policy that maximizes its expected utility:
\[\max_{\pi} f(\pi) \mbox{ subject to $\pi$ is feasible.}\]

\section{Algorithm Design}

We next describe our algorithm and analyze its performance. Our algorithm is based on the contention resolution scheme \cite{chekuri2014submodular}, which is proposed in the context of submodular maximization with deterministic item cost. We extend their design by considering state-dependent item cost and rejections. Our algorithm, called \textsf{StoCan}, is composed of two phases.

The first phase is done offline, we use the continuous greedy algorithm (Algorithm \ref{alg:greedy-peak}) to compute a fractional solution over a down monotone polytope.  The framework of  continuous greedy algorithm is first proposed by \cite{calinescu2011maximizing} in the context of submodular maximization subject to a matroid constraint. In particular,  Algorithm \ref{alg:greedy-peak} maintains an $I\times S$ matrix $\mathbf{y}(t)$, starting with $\mathbf{y}(0)=\mathbf{0}$. Let $R(t)$ contain each $(i,s)$ independently with probability $y_{is}(t)$.
For each $(i,s)\in [I]\times [S]$, estimate its weight $\omega_{is}$ as follows
\[\omega_{is}=\mathbb{E}[h(R(t)\cup\{(i,s)\})]-\mathbb{E}[h(R(t))]\]
For each pair of $i$ and $s$, let $ p_i(s)$ denote $\Pr[\Phi(i) = s]$ for short. Solve the following linear programming problem \textbf{LP} and obtain the optimal solution $\mathbf{x}^{LP}$, then update the fractional solution at round $t$ as $\forall (i,s)\in [I]\times[S], y_{is}(t+\delta)=y_{is}(t)+x^{LP}_{is}$.
\begin{center}
\framebox[0.42\textwidth][c]{
\enspace
\begin{minipage}[t]{0.42\textwidth}
\small
\textbf{LP:}
\emph{Maximize $\sum_{(i,s)\in [I]\times [S]}\omega_{is}x_{is}$}\\
\textbf{subject to:}
\begin{equation*}
\begin{cases}
\forall(i,s)\in [I]\times [S]: x_{is}\leq p_i(s) \\
  \sum_{(i,s)\in [I]\times [S]} x_{is}c_i(s)\leq B \\
\end{cases}
\end{equation*}
\end{minipage}}
\end{center}
\vspace{0.1in}

 After $1/\delta$ rounds, $\mathbf{y}(1/\delta)$ is returned as the final solution. In the rest of this paper, let $\mathbf{y}$ denote $\mathbf{y}(1/\delta)$ for short.

In the second phase, we implement a simple randomized policy based on $\mathbf{y}$. Our policy randomly picks a policy from  $\pi^{small}$ (Algorithm \ref{alg:greedy-peak1}) and $\pi^{large}$ (Algorithm \ref{alg:greedy-peak2}) with equal probability to execute. If  $\pi^{small}$  is picked, we discard all \emph{large} items whose cost is larger than $B/2$ (Line \ref{line:10} in Algorithm \ref{alg:greedy-peak1}), and add the rest of items according to the corresponding distribution in (scaled) $\mathbf{y}$ (Line \ref{line:7} in Algorithm \ref{alg:greedy-peak1}) . If  Algorithm $\pi^{large}$ is picked, we discard all \emph{small} items whose cost is no larger than $B/2$ (Line \ref{line:11} in Algorithm \ref{alg:greedy-peak2}), and add the rest of items according to the corresponding distribution in (scaled) $\mathbf{y}$  (Line \ref{line:8} in Algorithm \ref{alg:greedy-peak2}).


\begin{algorithm}[h]
{\small
\caption{Continuous Greedy}
\label{alg:greedy-peak}
\begin{algorithmic}[1]
\STATE Set $\delta=1/(IS)^2, t=0, f(\emptyset)=0, \mathbf{y}(0)=\mathbf{0}$.
\WHILE{$t<1$}
\STATE Let $R(t)$ contain each $(i,s)\in  [I]\times [S]$ independently with probability $y_{is}(t)$.
\STATE For each $(i,s)\in [I]\times [S]$, estimate
$\omega_{is}=\mathbb{E}[h(R(t)\cup\{(i,s)\})]-\mathbb{E}[h(R(t))]$;
\STATE Solve the following linear programming problem and obtain the optimal solution $\mathbf{x}^{LP}$
\STATE
\framebox[0.42\textwidth][c]{
\enspace
\begin{minipage}[t]{0.42\textwidth}
\small
\textbf{LP:}
\emph{Maximize $\sum_{(i,s)\in [I]\times [S]}\omega_{is}x_{is}$}\\
\textbf{subject to:}
\begin{equation*}
\begin{cases}
\forall(i,s)\in [I]\times [S]: x_{is}\leq p_i(s) \\
  \sum_{(i,s)\in [I]\times [S]} x_{is}c_i(s)\leq B \\
\end{cases}
\end{equation*}
\end{minipage}
}
\vspace{0.1in}
\STATE Let $y_{is}(t+\delta)=y_{is}(t)+x^{LP}_{is}$; \label{line:1}
\STATE Increment $t=t+\delta$;
\ENDWHILE
\RETURN $\mathbf{y}(1/\delta)$;
\end{algorithmic}
}
\end{algorithm}

\begin{algorithm}[h]
{\small
\caption{$\pi^{small}$}
\label{alg:greedy-peak1}
\begin{algorithmic}[1]
\STATE Set $G=\emptyset$, $i=1$.
\WHILE{$i\leq n$}
\STATE probe item $i$ and observe its state $s$
\IF {$c_i(s)> B/2$}
\STATE $i=i+1$; \COMMENT{discard all large items} \label{line:10}
\ELSE
\IF {the remaining budget is no less than $c_i(s)$} \label{line:17}
\STATE add $(i,s)$ to $G$ with probability $y_{is}/4p_i(s)$; \label{line:7}
\ENDIF
\STATE $i=i+1$;
\ENDIF
\ENDWHILE
\RETURN $G$;
\end{algorithmic}
}
\end{algorithm}

\begin{algorithm}[h]
{\small
\caption{$\pi^{large}$}
\label{alg:greedy-peak2}
\begin{algorithmic}[1]
\STATE Set $G=\emptyset$, $i=1$.
\WHILE{$i\leq n$}
\STATE probe item $i$ and observe its state $s$
\IF {$c_i(s)\leq B/2$}
\STATE $i=i+1$; \COMMENT{discard all small items} \label{line:11}
\ELSE
\IF {the remaining budget is no less than $c_i(s)$}
\STATE add $(i,s)$ to $G$ with probability $y_{is}/4p_i(s)$;  \label{line:8}
\ENDIF
\STATE $i=i+1$;
\ENDIF
\ENDWHILE
\RETURN $G$;
\end{algorithmic}
}
\end{algorithm}
We next provide the main theorem of this paper.
\begin{theorem}
\label{thm:1}
Let $\pi^*$ denote the optimal policy, the expected utility achieved by \textsf{StoCan} is at lest $\frac{1-1/e}{16}f(\pi^*)$.
\end{theorem}

Before presenting the proof of Theorem \ref{thm:1}, we first introduce four preparation lemmas.
\begin{lemma}
\label{lem:1}
Let $\mathbf{y}$ denote the fractional solution returned from Algorithm \ref{alg:greedy-peak}, $H(\mathbf{y})\geq (1-1/e)f(\pi^*)$.
\end{lemma}
\begin{proof}Given $\pi^*$, for each item-state pair $(i,s)\in [I]\times [S]$, let $y^*_{is}$ denote the probability that  $\Phi(i)=s$ and $i$ is picked by $\pi^*$.  Clearly, $\forall(i,s)\in [I]\times [S]:y^*_{is}\leq p_i(s)$. Moreover, consider a fixed realization $\phi$, for each $(i,s) \in[I]\times [S]$, let $\mathbf{1}_{i,s}$ be an indicator that $\phi(i)=s$ and $i$ is picked by $\pi^*$, we have $\sum_{(i,s)\in [I]\times [S]}\mathbf{1}_{i,s}c_i(s)\leq B$, Thus, \[\mathbb{E}[\sum_{(i,s)\in [I]\times [S]}\mathbf{1}_{i,s}c_i(s)]=\sum_{(i,s)\in [I]\times [S]}\mathbb{E}[\mathbf{1}_{i,s}]c_i(s)=\sum_{(i,s)\in [I]\times [S]}y^*_{is}c_i(s)\leq B\] where the expectation is taken over $\Phi$ with respect to $\mathcal{D}$. It follows that $\mathbf{y}^*$ is a feasible solution to \textbf{LP}.  Define $\mathbf{1}_{is}$ as the matrix that has a $1$ in the $(i,s)$-th entry and $0$ in all other entries. Let $h_V((i,s))=h(V\cup\{(i,s)\})-h(V)$ and $H_{\mathbf{y}(t)}((i,s))=H(\mathbf{y}(t)\vee \mathbf{1}_{is})-H(\mathbf{y}(t))$ denote the marginal utility of $(i,s)$ with respect to $V$ and $\mathbf{y}(t)$, respectively. We next bound the increment  of $H(\mathbf{y}(t))$ during one step of Algorithm \ref{alg:greedy-peak}.
\begin{eqnarray}f(\pi^*)&\leq& \min_{V\subseteq [I]\times [S]} \left(h(V)+\sum_{(i,s)\in [I]\times [S]} y^*_{is} h_V((i,s))\right)\label{eq:4}\\
&\leq& H(\mathbf{y}(t))+\sum_{(i,s)\in [I]\times [S]} y^*_{is} H_{\mathbf{y}(t)}((i,s))\\
&\leq& H(\mathbf{y}(t))+\sum_{(i,s)\in [I]\times [S]} x^{LP}_{is} H_{\mathbf{y}(t)}((i,s))
\end{eqnarray}
The first inequality is proved in \cite{calinescu2011maximizing}. The third inequality is due to $\mathbf{x}^{LP}$ is an optimal solution to \textbf{LP}. Then this lemma follows from the standard analysis on submodular maximization.
\end{proof}

Given the fractional solution $\mathbf{y}$ returned from Algorithm \ref{alg:greedy-peak}, we next introduce two new fractional solutions  $\overline{\mathbf{y}}$ and $\mathbf{\underline{y}}$.  Define $\overline{y}_{is}=y_{is}$ if $c_i(s)\leq B/2$, otherwise, $\overline{y}_{is}=0$. Define
$\underline{y}_{is}=y_{is}$ if $c_i(s)>B/2$, otherwise, $\underline{y}_{is}=0$. Due to the submodularity of $h$, we have the following lemma.
\begin{lemma}
\label{lem:4}
$H(\mathbf{\underline{y}})+H(\mathbf{\overline{y}})\geq H(\mathbf{y})$
\end{lemma}
We next bound the expected utility achieved by $\pi^{small}$.
\begin{lemma}
\label{lem:2}
$f(\pi^{small}) \geq H(\mathbf{\underline{y}})/8$
\end{lemma}
\begin{proof}Consider a modified version of $\pi^{small}$ by removing Line \ref{line:17}, that is, after probing an item $i$ and observing its state $s$, if $c_i(s)\leq B/2$, we select $i$ with probability $y_{is}/4p_i(s)$ regardless of the remaining budget, otherwise, we discard $i$. Denote by $G'$ the returned solution from the modified $\pi^{small}$. It is easy to verify that for each $(i,s)\in [I]\times [S]$ with  $c_i(s)\leq B/2$, the probability that  $(i,s)$ is included in $G'$ is $p_i(s)y_{is}/4p_i(s)=y_{is}/4$. Notice that since each item $i$ can only have one state, the event that $(i,s)$ is included in $G'$ is not independent from the event that $(i,s')$ is included in $G'$ where $s'$ is a different state from $s$ and $c_i(s')\leq B/2$. However, as shown in Lemma 3.7 in \cite{calinescu2011maximizing}, this dependency does not degrade  the expected utility, i.e.,  $\mathbb{E}[h(G')]\geq H(\mathbf{\underline{y}}/4)$. Due to $H$ is concave along any nonnegative direction \cite{calinescu2011maximizing}, we have $H(\mathbf{\underline{y}}/4) \geq H(\mathbf{\underline{y}})/4$. It follows that
\begin{equation}
\label{eq:1}
\mathbb{E}[h(G')]\geq H(\mathbf{\underline{y}}/4) \geq H(\mathbf{\underline{y}})/4
\end{equation}

 Next we focus on proving that
\begin{equation}
\label{eq:2}
f(\pi^{small}) = \mathbb{E}[h(G)]\geq \mathbb{E}[h(G')]/2
\end{equation}
This lemma follows from (\ref{eq:1}) and (\ref{eq:2}).

Recall that if the remaining budget is no less than $c_i(s)$, $\pi^{small}$ adds $(i,s)$ to $G$. Because $\mathbf{y}$ is a feasible solution to \textbf{LP}, $\mathbf{\underline{y}}$  is also  a feasible solution to \textbf{LP}, it implies that $\sum_{(i,s)\in [I]\times [S]} \underline{y}_{is}c_i(s)/4\leq B/4$. According to Markov's inequality, the probability that the remaining budget is less than $B/2$ is at most $1/2$. Because we assume $c_i(s)\leq B/2$,  the probability that the remaining budget is less than $c_i(s)$ is at most $1/2$. Thus,  the probability that $(i,s)$ is included in $G$ is at least $y_{is}/8$.

Let $G[i]$ (resp. $G'[i]$) denote all item-state pairs in $G$ (resp. $G'$) that involve items in $[i]$, i.e., $G[i]=G\cap \{(j,s)\mid j\in [i], s\in[S]\}$ and $G'[i]=G'\cap \{(j,s)\mid j\in [i], s\in[S]\}$. We next prove that for any $i\in[I]$,
\begin{equation}
\label{eq:3}
\mathbb{E}[h(G [i])-h(G [i-1])]\geq \frac{1}{2}\mathbb{E}[h(G' [i])-h(G' [i-1])]
\end{equation}
Notice that (\ref{eq:3}) implies (\ref{eq:2}) due to $\mathbb{E}[h(G)]= h(\emptyset)+\sum_{i=1}^n \mathbb{E}[h(G [i])-h(G [i-1])] \geq f(\emptyset)+\sum_{i=1}^n \frac{1}{2}\mathbb{E}[h(G' [i])-h(G' [i-1])]=\mathbb{E}[h(G')]/2$.

We first give an lower bound on $\mathbb{E}[h(G [i])-h(G [i-1])]$. For each $(i,s) \in[I]\times [S]$, let $\mathbf{1}_{(i,s)\in G}$ be the indicator that $(i,s)$ is included in $G$.
\begin{align*}
&\mathbb{E}[h(G [i])-h(G [i-1])] \\
&=\sum_{s\in [S]} \mathbb{E}\left[\mathbf{1}_{(i,s)\in G}\big(h(G [i-1]\cup (i,s))-h(G [i-1])\big)\right]\\
&\geq \sum_{s\in  [S]} \mathbb{E}\left[\mathbf{1}_{(i,s)\in G}\big(h(G' [i-1]\cup (i,s))-h(G' [i-1])\big)\right]\\
&\geq \sum_{s\in [S]} \mathbb{E}\left[\mathbf{1}_{(i,s)\in G}\right]\mathbb{E}\left[h(G' [i-1]\cup (i,s))-h(G' [i-1])\right]\\
&\geq \sum_{s\in [S]} \frac{y_{is}}{8}\mathbb{E}\left[h(G' [i-1]\cup (i,s))-h(G' [i-1])\right]
\end{align*}
The first inequality is due the submodularity of $f$. The second inequality follows from the same proof (monotonicity part) of Theorem 5 in \cite{fukunaga2019stochastic} and Assumption \ref{asum:1}. Moreover,
\begin{align*}
&\mathbb{E}[h(G' [i])-h(G' [i-1])] = \sum_{s\in [S]} \frac{y_{is}}{4} \mathbb{E}\left[h(G' [i-1]\cup (i,s))-h(G' [i-1])\right]
\end{align*}
Based on the above discussions, we have $\mathbb{E}[h(G [i])-h(G [i-1])]\geq \frac{1}{2}\mathbb{E}[h(G' [i])-h(G' [i-1])]$. This finishes the proof of (\ref{eq:3}) and hence (\ref{eq:2}).
\end{proof}

Now consider the second option $\pi^{large}$. In the following lemma, we prove that the expected utility achieved by $\pi^{large}$ is at least $H(\mathbf{\overline{y}})/8$.
\begin{lemma}
\label{lem:3}
$f(\pi^{large}) \geq H(\mathbf{\overline{y}})/8$.
\end{lemma}
\begin{proof}Because $\mathbf{y}$ is a feasible solution to \textbf{LP}, $\mathbf{\overline{y}}$  is also  a feasible solution to \textbf{LP}, it implies that $\sum_{(i,s)\in [I]\times [S]} \overline{y}_{is}c_i(s)/4\leq B/4$. Since we only consider those $(i,s)$ whose cost is larger than $B/2$, the probability that $G=\emptyset$  is at least $1/2$.
 Consider any $(i,s) \in[I]\times [S]$, conditioned on $G[i-1]=\emptyset$, the probability that $(i,s)$ is included in $G$ is at least $\overline{y}_{is}/4$. Thus, the probability that $(i,s)$ is included in $G$ is at least $\overline{y}_{is}/8$. Recall that $\pi^{large}$ only picks large items, $G$ contains at most one item (and its state) due to budget constraint. Thus, the expected utility of $\pi^{large}$ is at least $f(\pi^{large})\geq \sum_{(i,s)\in [I]\times [S]}\frac{\overline{y}_{is}h((i,s))}{8}$. Due to the submodularity of $h$ and Lemma 3.7 in \cite{calinescu2011maximizing}, we have $\sum_{(i,s)\in [I]\times [S]}\overline{y}_{is}h((i,s))/8 \geq  H(\mathbf{\overline{y}}/8)$. Since $H$ is concave along any nonnegative direction \cite{calinescu2011maximizing}, we have $H(\mathbf{\overline{y}}/8)\geq H(\mathbf{\overline{y}})/8$, thus $f(\pi^{large}) \geq H(\mathbf{\overline{y}})/8$. \end{proof}

 \paragraph{Proof of Theorem \ref{thm:1}:} Now we are ready to present the proof of Theorem \ref{thm:1}. Based on Lemma \ref{lem:2} and \ref{lem:3}, we have $f(\pi^{small})+f(\pi^{large})\geq \frac{H(\mathbf{\underline{y}})+ H(\mathbf{\overline{y}})}{8}$. This, together with Lemma \ref{lem:4}, implies that $f(\pi^{small})+f(\pi^{large})\geq \frac{ H(\mathbf{y})}{8}$. Because $H(\mathbf{y})\geq (1-1/e)f(\pi^*)$ as proved in Lemma \ref{lem:1}, we have $f(\pi^{small})+f(\pi^{large})\geq \frac{1-1/e}{8}f(\pi^*)$. Since  \textsf{StoCan} randomly picks one policy from  $\pi^{small}$ and $\pi^{large}$ to execute, the expected utility of  \textsf{StoCan} is at least $\frac{1-1/e}{16}f(\pi^*)$.
\section{Extension to Online Setting: A variant of Submodular Prophet Inequalities}
One nice feature about \textsf{StoCan} is that the implementation of $\pi^{small}$ and $\pi^{large}$ does not require any specific order of items. Therefore, \textsf{StoCan}   can be implemented in an online setting described as follows: suppose there is a sequence of items arriving with different states, on the arrival of an item, we observe its state and decide immediately and irrevocably whether to select it or not subject to a budget constraint. In this sense, the online version of our problem can be viewed as a variant of the submodular prophet inequalities \cite{rubinstein2017combinatorial,chekuri2021submodular}. Our setting differs from theirs in two ways: 1) our utility function is modeled as a lattice submodular function  over items and states, and 2) our model incorporates state-dependent cost. Similar to the offline setting,  \textsf{StoCan} first computes $\mathbf{y}$ using Algorithm \ref{alg:greedy-peak} in advance, then randomly picks one policy from $\pi^{small}$ and $\pi^{large}$ to execute. Notice that the online version of $\pi^{small}$ and $\pi^{large}$  probes the items in order of their arrival. It is easy to verify that this does not affect the performance analysis of \textsf{StoCan}, i.e., our analysis does not rely on any specific order of items, thus \textsf{StoCan} achieves  the same approximation ratio as obtained under the offline setting.
\section{Conclusion}
In this paper, we study the stochastic submodular probing problem with state-dependent costs and rejections. We present a constant approximate solution to this problem. We show that our solution can be implemented in an online fashion.
\bibliographystyle{ormsv080}
\bibliography{social-advertising-1}

\end{document}